\documentclass[pra,aps,nopacs,onecolumn,twoside,superscriptaddress]{revtex4}

\usepackage{multirow, makecell}
\usepackage{amsmath,amsfonts,amssymb,caption,color,epsfig,graphics,graphicx,hyperref,latexsym,mathrsfs,revsymb,theorem,url,verbatim,epstopdf,enumerate}
\usepackage{amsmath,amsfonts,amssymb,caption,hyperref,color,epsfig,graphics,graphicx,latexsym,mathrsfs,revsymb,theorem,url,verbatim,epstopdf,cleveref}
\usepackage{fontenc}
\usepackage{cases}
\usepackage[T1]{fontenc}
\let\polishl\l
\hypersetup{colorlinks,linkcolor={blue},citecolor={blue},urlcolor={red}}
\usepackage{lmodern} 
\usepackage{tipa}
\usepackage{textcomp} 

\newtheorem{definition}{Definition}
\newtheorem{proposition}[definition]{Proposition}
\newtheorem{lemma}[definition]{Lemma}

\newtheorem{theorem}[definition]{Theorem}
\newtheorem{corollary}[definition]{Corollary}
\newtheorem{conjecture}[definition]{Conjecture}

\newtheorem{remark}[definition]{Remark}
\newtheorem{example}[definition]{Example}
\newtheorem{question}[definition]{Question}
\newtheorem{memo}[definition]{Memo}

\def\squareforqed{\hbox{\rlap{$\sqcap$}$\sqcup$}}
\def\qed{\ifmmode\squareforqed\else{\unskip\nobreak\hfil
		\penalty50\hskip1em\null\nobreak\hfil\squareforqed
		\parfillskip=0pt\finalhyphendemerits=0\endgraf}\fi}
\def\endenv{\ifmmode\;\else{\unskip\nobreak\hfil
		\penalty50\hskip1em\null\nobreak\hfil\;
		\parfillskip=0pt\finalhyphendemerits=0\endgraf}\fi}
\newenvironment{proof}{\noindent \textbf{{Proof.~} }}{\qed}
\def\Dbar{\leavevmode\lower.6ex\hbox to 0pt
	{\hskip-.23ex\accent"16\hss}D}
\makeatletter
\def\url@leostyle{%
	\@ifundefined{selectfont}{\def\UrlFont{\sf}}{\def\UrlFont{\small\ttfamily}}}
\makeatother
\def\bcj{\begin{conjecture}}
	\def\ecj{\end{conjecture}}
\def\bcr{\begin{corollary}}
	\def\ecr{\end{corollary}}
\def\bd{\begin{definition}}
	\def\ed{\end{definition}}
\def\bea{\begin{eqnarray}}
	\def\eea{\end{eqnarray}}
\def\bem{\begin{enumerate}}
	\def\eem{\end{enumerate}}
\def\bex{\begin{example}}
	\def\eex{\end{example}}
\def\bim{\begin{itemize}}
	\def\eim{\end{itemize}}
\def\bl{\begin{lemma}}
	\def\el{\end{lemma}}
\def\bma{\begin{bmatrix}}
	\def\ema{\end{bmatrix}}
\def\bpf{\begin{proof}}
	\def\epf{\end{proof}}
\def\bpp{\begin{proposition}}
	\def\epp{\end{proposition}}
\def\bqu{\begin{question}}
	\def\equ{\end{question}}
\def\br{\begin{remark}}
	\def\er{\end{remark}}
\def\bt{\begin{theorem}}
	\def\et{\end{theorem}}
\def\bmm{\begin{memo}}
	\def\emm{\end{memo}}

\def\btb{\begin{tabular}}
	\def\etb{\end{tabular}}

	\newcommand{\nc}{\newcommand}
	
	\def\a{\alpha}
	\def\b{\beta}

	\def\l{\lambda}

	\def\s{\sigma}

	\def\G{\Gamma}

	\nc{\bbA}{\mathbb{A}} \nc{\bbB}{\mathbb{B}} \nc{\bbC}{\mathbb{C}}
	\nc{\bbD}{\mathbb{D}} \nc{\bbE}{\mathbb{E}} \nc{\bbF}{\mathbb{F}}
	\nc{\bbG}{\mathbb{G}} \nc{\bbH}{\mathbb{H}} \nc{\bbI}{\mathbb{I}}
	\nc{\bbJ}{\mathbb{J}} \nc{\bbK}{\mathbb{K}} \nc{\bbL}{\mathbb{L}}
	\nc{\bbM}{\mathbb{M}} \nc{\bbN}{\mathbb{N}} \nc{\bbO}{\mathbb{O}}
	\nc{\bbP}{\mathbb{P}} \nc{\bbQ}{\mathbb{Q}} \nc{\bbR}{\mathbb{R}}
	\nc{\bbS}{\mathbb{S}} \nc{\bbT}{\mathbb{T}} \nc{\bbU}{\mathbb{U}}
	\nc{\bbV}{\mathbb{V}} \nc{\bbW}{\mathbb{W}} \nc{\bbX}{\mathbb{X}}
	\nc{\bbZ}{\mathbb{Z}}
	
	\nc{\bA}{{\bf A}} \nc{\bB}{{\bf B}} \nc{\bC}{{\bf C}}
	\nc{\bD}{{\bf D}} \nc{\bE}{{\bf E}} \nc{\bF}{{\bf F}}
	\nc{\bG}{{\bf G}} \nc{\bH}{{\bf H}} \nc{\bI}{{\bf I}}
	\nc{\bJ}{{\bf J}} \nc{\bK}{{\bf K}} \nc{\bL}{{\bf L}}
	\nc{\bM}{{\bf M}} \nc{\bN}{{\bf N}} \nc{\bO}{{\bf O}}
	\nc{\bP}{{\bf P}} \nc{\bQ}{{\bf Q}} \nc{\bR}{{\bf R}}
	\nc{\bS}{{\bf S}} \nc{\bT}{{\bf T}} \nc{\bU}{{\bf U}}
	\nc{\bV}{{\bf V}} \nc{\bW}{{\bf W}} \nc{\bX}{{\bf X}}
	\nc{\bZ}{{\bf Z}}
	
	\nc{\as}{{\cal AS}}
	\nc{\app}{{\cal AP}}
	\nc{\ar}{{\cal AR}}
	\nc{\bp}{{\cal BP}}
	\nc{\dbp}{{\cal DBP}}
\nc{\ew}{{\cal EW}}
\nc{\dew}{{\cal DEW}}
\nc{\ndew}{{\cal NDEW}}
\nc{\conv}{{\text{Conv}}}

	\nc{\cA}{{\cal A}} \nc{\cB}{{\cal B}} \nc{\cC}{{\cal C}}
	\nc{\cD}{{\cal D}} \nc{\cE}{{\cal E}} \nc{\cF}{{\cal F}}
	\nc{\cG}{{\cal G}} \nc{\cH}{{\cal H}} \nc{\cI}{{\cal I}}
	\nc{\cJ}{{\cal J}} \nc{\cK}{{\cal K}} \nc{\cL}{{\cal L}}
	\nc{\cM}{{\cal M}} \nc{\cN}{{\cal N}} \nc{\cO}{{\cal O}}
	\nc{\cP}{{\cal P}} \nc{\cQ}{{\cal Q}} \nc{\cR}{{\cal R}}
	\nc{\cS}{{\cal S}} \nc{\cT}{{\cal T}} \nc{\cU}{{\cal U}}
	\nc{\cV}{{\cal V}} \nc{\cW}{{\cal W}} \nc{\cX}{{\cal X}}
	\nc{\cZ}{{\cal Z}}
	
	\nc{\cpp}{{\cal PP}}
	
	\nc{\hA}{{\hat{A}}} \nc{\hB}{{\hat{B}}} \nc{\hC}{{\hat{C}}}
	\nc{\hD}{{\hat{D}}} \nc{\hE}{{\hat{E}}} \nc{\hF}{{\hat{F}}}
	\nc{\hG}{{\hat{G}}} \nc{\hH}{{\hat{H}}} \nc{\hI}{{\hat{I}}}
	\nc{\hJ}{{\hat{J}}} \nc{\hK}{{\hat{K}}} \nc{\hL}{{\hat{L}}}
	\nc{\hM}{{\hat{M}}} \nc{\hN}{{\hat{N}}} \nc{\hO}{{\hat{O}}}
	\nc{\hP}{{\hat{P}}} \nc{\hR}{{\hat{R}}} \nc{\hS}{{\hat{S}}}
	\nc{\hT}{{\hat{T}}} \nc{\hU}{{\hat{U}}} \nc{\hV}{{\hat{V}}}
	\nc{\hW}{{\hat{W}}} \nc{\hX}{{\hat{X}}} \nc{\hZ}{{\hat{Z}}}
	
	\nc{\hn}{{\hat{n}}}
	
	\def\dim{\mathop{\rm Dim}}

	\def\max{\mathop{\rm max}}

	\def\rank{\mathop{\rm rank}}
	
	\newcommand{\bra}[1]{\langle#1|}
	\newcommand{\ket}[1]{|#1\rangle}

	\newcommand{\norm}[1]{\lVert#1\rVert}

	\def\Dbar{\leavevmode\lower.6ex\hbox to 0pt
		{\hskip-.23ex\accent"16\hss}D}

\begin{document}

\providecommand{\Tr}{\operatorname{Tr}}
\providecommand{\rank}{\operatorname{rank}}
\providecommand{\Sym}{\operatorname{Sym}}
\providecommand{\Ran}{\operatorname{Ran}}
\providecommand{\ip}[2]{\langle #1|#2\rangle}

\title{A partial-trace matrix inequality and Werner-state distillability}

\author{Zhiwei Song}\email[]{zhiweisong@cuhk.edu.cn}
\affiliation{School of Data Science, The Chinese University of Hong Kong, Shenzhen,
Guangdong, 518172, China}

\author{Lin Chen}\email[]{linchen@buaa.edu.cn}
\affiliation{LMIB(Beihang University), Ministry of education, and School of Mathematical Sciences, Beihang University, Beijing 100191, China}

\begin{abstract}
Motivated by the equivalent partial-trace formulations of Werner-state distillability 
[P.~Costa Rico, Lett. Math. Phys. \textbf{115}, 47 (2025); 
S.-Y.~Qi \textit{et al.}, Phys. Rev. A \textbf{110}, 012406 (2024)], 
we prove a bipartite partial-trace inequality for every matrix of rank at most two. As applications, we prove  the two-copy undistillability of NPT Werner states in arbitrary local dimension, thereby resolving this open problem highlighted in 
[P.~Horodecki \textit{et al.}, PRX Quantum \textbf{3}, 010101 (2022)]. 
We further prove a two-parameter extension of the matrix inequality and show that two individually one-copy-undistillable NPT Werner states cannot activate each other's one-copy distillability. We also resolve the singular-value maximization problem associated with the two-ququart case.
\end{abstract}

\maketitle


\section{Introduction}

Pure entanglement plays a key role in various quantum-information tasks outperforming the classical counterparts such as teleportation \cite{teleportation1993} and error correction \cite{PhysRevA.54.3824}. Pure entanglement inevitably evolves into a mixed state due to the noise in a lab. Hence, the task of distilling pure entanglement from mixed states is a fundamental task in quantum information theory. The states are called distillable when the task is feasible. That is, a nonzero amount of pure entanglement can be asymptotically extracted from many copies of mixed states through local operations and classical communication (LOCC) \cite{rains1999,rains2001,br03}. Distillable entanglement is related to many issues of quantum information and physics, such as distinguishability of Bell states \cite{Ghosh2001DistinguishabilityOB}, entanglement negativity \cite{Vidal2002ComputableMO}, thermodynamically quantifying quantum correlations \cite{Oppenheim2002ThermodynamicalAT}, quantum communication through an unmodulated spin chain \cite{Bose2003QuantumCT}, entangling power of passive optical elements \cite{Wolf2003EntanglingPO}, squashed entanglement \cite{Christandl2004Squashed}, the generation of security keys \cite{dw2005}, entanglement in noninertial frames \cite{FuentesSchuller2005AliceFI}, capacity of quantum channels \cite{Devetak2005Capacity}, extremality of Gaussian states \cite{Wolf2006Extremality}, random states in high-dimensional bipartite systems \cite{Hayden2006AspectsOG}, discord in measurements \cite{Streltsov2011LinkingQD}, lattice gauge theories \cite{PhysRevLett.117.131602}, resource theory of coherence \cite{Chitambar2016RelatingTR} and so on. Recent experiments indicate that distillation protocols with limited resources are operationally relevant, especially in photonic and high-dimensional platforms where multiple copies and collective operations are costly \cite{EckerSingleCopy}. In particular, two-qutrit Werner states have been prepared experimentally, and stochastic local filters have been used to reveal hidden useful quantum features \cite{FangWernerExperiment}, providing further motivation for the finite-copy distillability problem.

It has been shown that distillable states have 
non-positive-partial transpose (NPT) 
\cite{Horodecki1996Separability,peres1996}. In contrast, PPT bound entangled states have been constructed in \cite{horodecki1997,PhysRevLett.80.5239}. The long-standing distillability problem asks whether all NPT entangled states are distillable and become available resources. The problem has attracted considerable attentions, for a review see \cite{Horodecki_2009}. The distillability problem has been investigated using semidefinite programming \cite{vd06}. However for any integer $n$, there are $n$-copy undistillable but $(n+1)$-copy distillable entangled states \cite{Watrous2004ManyCopies}. So the distillability problem seriously undermines the idea of computer verification. Recently, the problem has been seen as one of the five fundamental problems in the scope of theoretical quantum information \cite{PRXQuantum.3.010101}. It is widely believed that the answer to the problem might be negative, and the existence of NPT bound entangled states would imply the non-additivity of distillable entanglement and activation of bound entanglement \cite{Shor2000NonadditivityOB,WOS:000306933000001}. It also would manifest the irreversibility in asymptotic manipulations of entanglement \cite{Vidal2001Irreversibility}.

Technically, qubit-qudit NPT states are proven to be distillable \cite{hhh97}. The distillable entanglement can also be verified via the violation of reduction criterion and creation of entanglement witness \cite{hh1999,klc02}. From a mathematical point of view, the rank of a target NPT state heavily determines its distillability \cite{cd11jpa}, and the bipartite NPT states with rank at most four turn out to be 1-distillable \cite{cc08,cd16pra}. 
Further, the distillbility problem has turned out to be closely related to the well-known family of Werner states on $\bbC^d\otimes\bbC^d$
\begin{eqnarray}\label{eq:werner}
 \rho_{\alpha,d}=\frac{I_{d^2}+\alpha F_d}{d^2+\alpha d},
 \qquad -1\le\alpha\le1,
\end{eqnarray}
where $F_d(\ket{x}\otimes\ket{y})=\ket{y}\otimes\ket{x}$ denotes the swap operator for $\ket{x},\ket{y}\in\bbC^d$. The known distillability picture for the parameter $\alpha$ can be summarized as follows. The state $\rho_{\alpha,d}$ is separable, equivalently PPT, exactly when $\alpha\ge -1/d$ \cite{werner89}. It is already one-copy distillable when $\alpha\in [-1,-1/2)$ \cite{dcl00,PhysRevA.61.062312}. Thus, the remaining problem is to decide the $n(\ge2)$-copy distillability of  NPT Werner states for $\a\in[-1/2,-1/d)$. 
This is actually equivalent to the distillability of $n(\ge2)$-copy $\rho_{-1/2,d}$, shown in 2000 \cite{PhysRevA.61.062312}. A special case of this problem was presented in 2007. That is, Ref. \cite{5508622} converted the distillability problem of two-copy $4\times4$ Werner states into the maximization of square sum of two singular values, with some progress made in the past years \cite{QIAN2021139,SIO2025152}. Recently, Ref. \cite[Theorem~1]{2025New} and Ref. \cite[Theorem V.1]{qi2024nonpositive} have respectively reduced the distillability problem to the following multipartite matrix inequality in terms of partial trace.

\begin{theorem}\label{thm:costa-rico}
Let $\mathcal H$ be a finite-dimensional Hilbert space with a tensor decomposition $\mathcal H=\mathcal H_1\otimes\cdots\otimes\mathcal H_n$, where $\dim\mathcal H_i=d_i$. For $\alpha\in\bbR$ and $C\in\mathcal L(\mathcal H)$, define
\begin{eqnarray}
\label{eq:q^(n)}
 q^{(n)}_\alpha(C)
 :=\sum_{J\in\mathcal P(\{1,2,\ldots,n\})}\alpha^{|J|}\norm{\Tr_J C}_F^2,
\end{eqnarray}
where $\mathcal P(\{1,2,\ldots,n\})$ is the power set, $\Tr_J$ denotes the partial trace over the tensor factors indexed by $J$, and $\Tr_\emptyset C=C$. Then for $d\ge2$ and $-1\le\alpha\le1$, the Werner state $\rho_{\alpha,d}$ is $n$-copy distillable if and only if there exists $C\in\mathcal L((\bbC^d)^{\otimes n})$ with $\rank C\le2$ such that $q^{(n)}_\alpha(C)<0$.
\end{theorem}

We use the bipartite Hilbert space $\mathcal H=\mathcal A\otimes\mathcal B$ with $\dim\mathcal A=m$ and $\dim\mathcal B=n$. 
According to Theorem~\ref{thm:costa-rico}, if we can prove that for every $C\in\mathcal L(\mathcal H)$ with $\rank C\le2$, the following identity 
\begin{eqnarray}\label{eq:qalpha}
 q^{(2)}_{-\frac 12}(C)
 =\norm C_F^2
 -\frac 12(\norm{\Tr_{\mathcal A}C}_F^2+\norm{\Tr_{\mathcal B}C}_F^2)
 +\frac 14 |\Tr C|^2
\end{eqnarray}
is always non-negative, then we can prove the two-copy undistillability of the NPT Werner state.

In this paper, we establish the sharp rank-two partial-trace inequality in Theorem~\ref{thm:main}.  The proof converts the quadratic form into an explicit projection decomposition on two copies of the underlying bipartite Hilbert space; see Lemma~\ref{lem:projection}.  The two-copy distillability threshold of Werner states in arbitrary local dimension then follows as a direct application.  We also derive in Corollary~\ref{thm:two-parameter} a two-parameter extension excluding mutual one-copy activation for two independent NPT Werner states.  As another application, Corollary~\ref{cor:d4} resolves the singular-value maximization problem associated with the two-copy $4\times4$ Werner case proposed in Refs.~\cite{5508622,QIAN2021139}.

\section{Preliminaries}
In this section, we introduce some notations and preliminary results. For a finite-dimensional Hilbert space $\mathcal K$, we write $\mathcal L(\mathcal K)$ for the space of linear operators on $\mathcal K$.
All Hilbert-space inner products are linear in the second variable. For a matrix $X$, we denote
$\sigma_1(X)\ge\sigma_2(X)\ge\cdots\ge0$
as its singular values in non-increasing order.  We write $\mathcal M_{m,n}$ for the space of complex $m\times n$ matrices, and simply write $\mathcal M_n$ for $m=n$.  We identify $\mathcal M_{m,n}$ with $\mathcal H$ by
\begin{eqnarray}\label{eq:vectorization}
 X=[X_{ij}]\longmapsto
 \ket{x}:=\sum_{i,j}X_{ij}\ket{i}_{\mathcal A}\otimes\ket{j}_{\mathcal B},
\end{eqnarray}
where $\ket{i}_{\mathcal A}$ and $\ket{j}_{\mathcal B}$ are orthonormal bases of $\mathcal A$ and $\mathcal B$, respectively. Under this identification,
$
 \ip{x}{y}=\Tr(X^*Y)$ holds for $X,Y\in \mathcal M_{m,n}$.
We use lowercase letters for vectors in $\mathcal H$ and the corresponding uppercase letters for their matrix representatives in $\mathcal M_{m,n}$. We denote by $P_{\mathcal D}$ the projection onto a space $\mathcal D$. For a subspace $\mathcal K\subseteq\mathcal H$, let $\{\ket{f_i}\}$ be an orthonormal basis of $\mathcal K$. We denote by $\Sym^2\mathcal K$ the symmetric subspace of $\mathcal K\otimes\mathcal K$, namely 
\begin{eqnarray} \Sym^2\mathcal K :=\operatorname{Span}\{
 \frac{\ket{f_i}\otimes\ket{f_j}+\ket{f_j}\otimes\ket{f_i}}{2}: i\le j
\}.
\end{eqnarray}

On $\mathcal H\otimes\mathcal H$, let $F_{\mathcal A}$ and $F_{\mathcal B}$ exchange the two copies of $\mathcal A$ and $\mathcal B$, respectively, namely
\begin{eqnarray}\label{eq:partial-swaps-action}
 F_{\mathcal A}((\ket a\otimes\ket b)\otimes(\ket{a'}\otimes\ket{b'}))
 &=&(\ket{a'}\otimes\ket b)\otimes(\ket a\otimes\ket{b'}),\nonumber\\
 F_{\mathcal B}((\ket a\otimes\ket b)\otimes(\ket{a'}\otimes\ket{b'}))
 &=&(\ket a\otimes\ket{b'})\otimes(\ket{a'}\otimes\ket b).
\end{eqnarray}
By definition, the two swaps operators $F_{\mathcal A}$ and $F_{\mathcal B}$ commute with each other. Moreover, $F_{\mathcal A}F_{\mathcal B}$ is the full swap of the two copies of $\mathcal H$. 
We set
\begin{eqnarray}
\label{eq:projections}
P_{\mathcal A}^{\pm}&:=\frac{I\pm F_{\mathcal A}}{2},\qquad
P_{\mathcal B}^{\pm}:=\frac{I\pm F_{\mathcal B}}{2},\nonumber\\
P_{++}&:=P_{\mathcal A}^{+}P_{\mathcal B}^{+}
=\frac{I+F_{\mathcal A}+F_{\mathcal B}+F_{\mathcal A}F_{\mathcal B}}{4},
\nonumber\\
P_{--}&:=P_{\mathcal A}^{-}P_{\mathcal B}^{-}
=\frac{I-F_{\mathcal A}-F_{\mathcal B}+F_{\mathcal A}F_{\mathcal B}}{4}.
\end{eqnarray}
We have $P_{\mathcal A}^{\pm}$ and $P_{\mathcal B}^{\pm}$ are mutually commuting projections. Consequently, $P_{++}$ and $P_{--}$ are projections with
\begin{eqnarray}\label{eq:parity-ranges}
 \Ran P_{++}
 &=&\{\ket{v}:F_{\mathcal A}\ket{v}=\ket{v},\ F_{\mathcal B}\ket{v}=\ket{v}\},\nonumber\\
 \Ran P_{--}
 &=&\{\ket{v}:F_{\mathcal A}\ket{v}=-\ket{v},\ F_{\mathcal B}\ket{v}=-\ket{v}\}.
\end{eqnarray}
Hence, with $P_{\Sym^2\mathcal H}:=(I+F_{\mathcal A}F_{\mathcal B})/2$, the full-swap symmetric subspace decomposes as
\begin{eqnarray}\label{eq:full-symmetric-decomposition}
 \Ran P_{\Sym^2\mathcal H}
 =\Ran P_{++}\oplus\Ran P_{--}.
\end{eqnarray}

For the vectors $\ket{a},\ket{b},\ket{a'},\ket{b'}\in\mathcal H$ corresponding to $A,B,A',B'\in \mathcal M_{m,n}$, the following identities hold:
\begin{eqnarray}\label{eq:swap-A}
 \bra{a}\otimes \bra{b}F_{\mathcal A}\ket{a'}\otimes\ket{b'}=\Tr(A^*B'B^*A'),
\end{eqnarray}
\begin{eqnarray}\label{eq:swap-B}
 \bra{a}\otimes \bra{b}F_{\mathcal B}\ket{a'}\otimes\ket{b'}=\Tr(A^*A'B^*B'),
\end{eqnarray}
\begin{eqnarray}\label{eq:swap-full}
 \bra{a}\otimes \bra{b}F_{\mathcal A}F_{\mathcal B}\ket{a'}\otimes\ket{b'}
 =\ip{a}{b'}\ip{b}{a'}.
\end{eqnarray}
Indeed, writing $\ket a=\sum_{i,\alpha}A_{i\alpha}\ket{i}_{\mathcal A}\otimes\ket{\alpha}_{\mathcal B}$, $\ket b=\sum_{j,\beta}B_{j\beta}\ket{j}_{\mathcal A}\otimes\ket{\beta}_{\mathcal B}$, and similarly for $\ket{a'}$ and $\ket{b'}$, we have
\begin{eqnarray}
 \bra{a,b}F_{\mathcal A}\ket{a',b'}
 &=&\sum_{i,\alpha,j,\beta}\overline{A_{i\alpha}}\,\overline{B_{j\beta}}\,A'_{j\alpha}B'_{i\beta}\nonumber\\
 &=&\sum_{\alpha,\beta}(A^*B')_{\alpha\beta}(B^*A')_{\beta\alpha}
 =\Tr(A^*B'B^*A').
\end{eqnarray}
The proof of \eqref{eq:swap-B} is similar, and \eqref{eq:swap-full} follows because $F_{\mathcal A}F_{\mathcal B}\ket{a'}\otimes\ket{b'}=\ket{b'}\otimes\ket{a'}$.

Let $C\in\mathcal L(\mathcal H)$ have rank at most two. There exist orthonormal vectors $\ket{e_1},\ket{e_2}\in\mathcal H$ whose span contains $\Ran C^*$. Let $\ket{x_i}=C\ket{e_i}$. Since $C$ vanishes on $(\Ran C^*)^\perp$, we have
\begin{eqnarray}\label{eq:two-column}
 C=\ket{x_1}\bra{e_1}+\ket{x_2}\bra{e_2}.
\end{eqnarray}
When $C$ has rank one, we may choose $\ket{x_2}=0$. Let
\begin{eqnarray}\label{eq:S-w}
 \mathcal S:=\operatorname{span}\{\ket{e_1},\ket{e_2}\},\qquad
 \ket{w}:=\ket{e_1}\otimes\ket{x_2}-\ket{e_2}\otimes\ket{x_1}.
\end{eqnarray}
The following proposition presents an equivalent form of $q^{(2)}_{-1/2}(C)$ in \eqref{eq:qalpha}, with proof given in Appendix \ref{g}. 

\begin{proposition}\label{prop:three-square}
For any $C\in\mathcal L(\bbC^m\otimes\bbC^n)$ writing in the form (\ref{eq:two-column}),
the quadratic form in \eqref{eq:qalpha} satisfies
\begin{eqnarray}\label{eq:three-square}
 q^{(2)}_{-1/2}(C)
 &=&2\norm{P_{--}(\ket{e_1}\otimes\ket{x_1})}^2
 +\norm{P_{--}(\ket{e_1}\otimes\ket{x_2}+\ket{e_2}\otimes\ket{x_1})}^2\nonumber\\
 &+&2\norm{P_{--}(\ket{e_2}\otimes\ket{x_2})}^2
 +\norm{P_{++}\ket{w}}^2
 -\frac12\norm{P_{\Sym^2\mathcal S}\ket{w}}^2.
\end{eqnarray}
\end{proposition}

\section{Proof of the main theorem}
In this section, we prove the main inequality and derive its immediate consequences. It suffices to prove that the sum of the last two terms in \eqref{eq:three-square} is nonnegative. We need the following lemma.

\begin{lemma}\label{lem:occupation}
Let $Q$ be a rank-two orthogonal projection on $\mathcal H$. Suppose $\ket{z}\in\Ran P_{--}$ and
\begin{eqnarray}\label{eq:z-coefficient}
 \ket{z}=\sum_{i,j}Z_{ij}\ket{h_i}\otimes\ket{h_j}
\end{eqnarray}
for an orthonormal basis $\ket{h_i}$ of $\mathcal H$. Then
\begin{eqnarray}\label{eq:occupation}
 \ip{z}{(Q\otimes I+I\otimes Q)z}
 \le2(\sigma_1(Z)^2+\sigma_2(Z)^2)
 \le\norm Z_F^2=\norm{\ket{z}}^2.
\end{eqnarray}
\end{lemma}

\begin{proof}
Since $\ket{z}\in\Ran P_{--}$, we have  
$P_{--}\ket{z}=\ket{z}$ and
$F_{\mathcal A}\ket{z}=F_{\mathcal B}\ket{z}=-\ket{z}$. Thus $\ket{z}$ is invariant under the full swap $F_{\mathcal A}F_{\mathcal B}$. So $
 Z^T=Z
$. It follows from \eqref{eq:z-coefficient} $\ip{z}{(Q\otimes I)z}=\Tr(QZZ^*)$ and $\ip{z}{(I\otimes Q)z}=\Tr(Q(Z^*Z)^T)$. Consequently, Ky Fan's maximum principle implies that
\begin{eqnarray}
	\label{eq:ky-fan}
 \ip{z}{(Q\otimes I+I\otimes Q)z}=2\Tr(QZZ^*)
 \le2(\sigma_1(Z)^2+\sigma_2(Z)^2).
 \end{eqnarray}
Using the results of \cite[Corollary 4.4.4(c)]{HornJohnson2013}(see also \cite{Takagi1924}), there are orthonormal vectors $\ket{u_p}\in\mathcal H$ such that
\begin{eqnarray}\label{eq:takagi-z}
 \ket{z}=\sum_p\sigma_p(Z)\ket{u_p}\otimes\ket{u_p}.
\end{eqnarray}
We set
$
 \ket{z_2}:=\sum_{p=1}^2\sigma_p(Z)\ket{u_p}\otimes\ket{u_p}.
$
Then
\begin{eqnarray}\label{eq:two-singular-values1}
 \sigma_1(Z)^2+\sigma_2(Z)^2&=&\norm{\ket{z_2}}^2=\ip{z_2}{z}
 =\ip{P_{--}z_2}{z}
 \le\norm{P_{--}\ket{z_2}}\norm{\ket{z}}\nonumber\\
 &=&\sqrt{\bigg(\frac12\norm{\ket{z_2}}^2
 -\frac14\ip{z_2}{(F_{\mathcal A}+F_{\mathcal B})z_2}\bigg)}\norm{\ket{z}}
 \end{eqnarray}
Let $U_p$ be the matrix representative of $\ket{u_p}$ under \eqref{eq:vectorization}. Recalling from \eqref{eq:swap-A} and \eqref{eq:swap-B}, we obtain that
\begin{eqnarray}\label{eq:z2-swap}
 \frac12\ip{z_2}{(F_{\mathcal A}+F_{\mathcal B})z_2}
 =\sum_{p=1}^2\sigma_p(Z)^2\Tr((U_p^*U_p)^2)
 +2\sigma_1(Z)\sigma_2(Z)\Re\Tr((U_1^*U_2)^2).
\end{eqnarray}
Using $|\Tr(R^2)|\le\Tr(R^*R)$ with $R=U_1^*U_2$, and applying the Cauchy--Schwarz inequality, we have
\begin{eqnarray}\label{eq:cross-bound}
 |\Tr((U_1^*U_2)^2)|
 \le\Tr(U_1U_1^*U_2U_2^*)
 \le\sqrt{\Tr((U_1^*U_1)^2)\Tr((U_2^*U_2)^2)}.
\end{eqnarray}
Combining \eqref{eq:z2-swap} and \eqref{eq:cross-bound}, we obtain
\begin{eqnarray}\label{eq:swap-positive}
 \frac12\ip{z_2}{(F_{\mathcal A}+F_{\mathcal B})z_2}
 \ge\biggl(\sigma_1(Z)\sqrt{\Tr((U_1^*U_1)^2)}-\sigma_2(Z)\sqrt{\Tr((U_2^*U_2)^2)}\biggl)^2
 \ge 0.
\end{eqnarray}
Taking back to \eqref{eq:two-singular-values1}, we have
$\norm{\ket{z_2}}^2\le \frac1{\sqrt2}\norm{\ket{z_2}}\norm{\ket{z}}\le \frac12\norm{\ket{z}}^2=\frac12\norm Z_F^2$. Combining with \eqref{eq:ky-fan}, 
this completes the proof.
\end{proof}

\begin{lemma}\label{lem:projection}
Let $\mathcal R\subset\mathcal H$ be a two-dimensional subspace. For every $\ket{\xi}\in\mathcal R\otimes\mathcal H$,
\begin{eqnarray}\label{eq:projection-estimate}
 \norm{P_{\Sym^2\mathcal R}\ket{\xi}}^2
 \le2\norm{P_{++}\ket{\xi}}^2.
\end{eqnarray}
\end{lemma}

\begin{proof}
Let $\ket{\xi_0}=P_{\Sym^2\mathcal H}\ket{\xi}$ and $\ket{y}=P_{\Sym^2\mathcal R}\ket{\xi}$. 
Since $\ket{\xi}\in\mathcal R\otimes\mathcal H$, the symmetric vector $\ket{\xi_0}$ has no component in $\Sym^2(\mathcal R^\perp)$. Hence it has an orthogonal decomposition
\begin{eqnarray}\label{eq:w0-decomposition}
 \ket{\xi_0}=\ket{y}+\ket{u},
\end{eqnarray}
where $\ket{u}\in(\mathcal R\otimes\mathcal R^\perp)\oplus(\mathcal R^\perp\otimes\mathcal R)$.
On the other hand, by \eqref{eq:full-symmetric-decomposition}, we have $\ket{\xi_0}=P_{++}\ket{\xi_0}+P_{--}\ket{\xi_0}$. Since $\Ran P_{++}\subseteq\Ran P_{\Sym^2\mathcal H}$, we also have $P_{++}\ket{\xi_0}=P_{++}\ket{\xi}$. Therefore, 
\begin{eqnarray}
	\label{eq:projection-estimate-1}
 \norm{P_{++}\ket{\xi}}^2
 =\norm{\ket{\xi_0}}^2-\norm{P_{--}\ket{\xi_0}}^2.
\end{eqnarray}

Define $N:=P_{\mathcal R}\otimes I+I\otimes P_{\mathcal R}$. One can verify that $N\ket{y}=2\ket{y}$ and $N\ket{u}=\ket{u}$.
Therefore, 
\begin{eqnarray}\label{eq:Pmm-w0}
 \norm{P_{--}\ket{\xi_0}}^2
 &=&\ip{P_{--}\xi_0}{\xi_0}
 =\ip{P_{--}\xi_0}{N(\frac12y+u)}\nonumber\\
 &\le&\sqrt{\ip{P_{--}\xi_0}{NP_{--}\xi_0}\ip{\frac12y+u}{N(\frac12y+u)}}\nonumber\\
 &\le&\norm{\ket{P_{--}\xi_0}}\sqrt{\frac12\norm{\ket{y}}^2+\norm{\ket{u}}^2},
 \end{eqnarray}
where the first inequality follows from the Cauchy--Schwarz inequality, and the second follows from Lemma~\ref{lem:occupation}. 
This yields
$
 \norm{P_{--}\ket{\xi_0}}^2
 \le\frac12\norm{\ket{y}}^2+\norm{\ket{u}}^2.
$
Combining with (\ref{eq:w0-decomposition}) and (\ref{eq:projection-estimate-1}), we obtain $\norm{P_{++}\ket{\xi}}^2 \ge\frac12\norm{\ket{y}}^2$.
This completes the proof.
\end{proof}

Based on the preceding lemmas, we now prove the main theorem.
\begin{theorem}\label{thm:main}
If $C\in\mathcal L(\bbC^m\otimes\bbC^n)$ and $\rank C\le2$, then $q^{(2)}_{-1/2}(C)$ in \eqref{eq:qalpha} is nonnegative, equivalently,
\begin{eqnarray}\label{eq:main}
 \norm{\Tr_{\mathcal A} C}_F^2+\norm{\Tr_{\mathcal B} C}_F^2
 \le 2\norm C_F^2+\frac12|\Tr C|^2.
\end{eqnarray}
Consequently, for every $d\ge2$ and $-1\le\alpha\le1$, the Werner state $\rho_{\alpha,d}$ defined in \eqref{eq:werner} is two-copy undistillable if and only if $\alpha\ge-1/2$. In particular, for $d\ge3$, every Werner state with $-1/2\le\alpha<-1/d$ is NPT and two-copy undistillable.
\end{theorem}

\begin{proof}
Consider the form of $C$ in \eqref{eq:three-square},
by applying Lemma~\ref{lem:projection} to $\ket{w}$, the right-hand side is nonnegative. This proves $q^{(2)}_{-1/2}(C)\ge0$, which is equivalent to \eqref{eq:main}.

The consequent two-copy undistillability of Werner states follows from the established equivalences explained in the Introduction. This completes the proof.
\end{proof}

As an application, the inequality also proves the conjectured matrix inequality arising from the two-copy $4\times4$ Werner problem in \cite{5508622} and further studied in \cite{QIAN2021139,SIO2025152}.
\begin{corollary}
\label{cor:d4}
Let $d\ge4$ and $\s_j(X)$ the $j$'th singular value of $X$. If $A,B\in \mathcal{M}_d$ are traceless, then
\begin{eqnarray}\label{eq:d4}
 \sigma_1(A\otimes I_d+I_d\otimes B)^2
 +\sigma_2(A\otimes I_d+I_d\otimes B)^2
 \le\frac{3d-4}{d}(\norm{A}_F^2+\norm{B}_F^2).
\end{eqnarray}
\end{corollary}
\begin{proof}
For every matrix $X$, we have
\begin{eqnarray}\label{eq:kyfan-duality}
 (\sigma_1(X)^2+\sigma_2(X)^2)^{1/2}
 =\max_{\rank C\le2,\ \norm C_F=1}|\Tr(C^*X)|.
\end{eqnarray}
Let $X=A\otimes I_d+I_d\otimes B$. Since $A$ and $B$ are traceless,
we have
\begin{eqnarray}\label{eq:traceless-pairing}
 \Tr(C^*X)
 &=&\ip{\Tr_{\mathcal B}C-\frac{\Tr C}{d}I_d}{A}
 +\ip{\Tr_{\mathcal A}C-\frac{\Tr C}{d}I_d}{B}.
\end{eqnarray}
Applying Theorem~\ref{thm:main}, we have
\begin{eqnarray}\label{eq:traceless-partial-traces}
&&\lefteqn{\norm{\Tr_{\mathcal A}C-\frac{\Tr C}{d}I_d}_F^2
+\norm{\Tr_{\mathcal B}C-\frac{\Tr C}{d}I_d}_F^2}
\nonumber\\
&=&\norm{\Tr_{\mathcal A}C}_F^2+\norm{\Tr_{\mathcal B}C}_F^2-\frac{2}{d}|\Tr C|^2\le 2\norm C_F^2+\left(\frac12-\frac2d\right)|\Tr C|^2\le \frac{3d-4}{d}\norm C_F^2,
\end{eqnarray}
where the last inequality uses the fact that $|\Tr C|^2\le \rank(C)\norm C_F^2$, see \cite[Chapter~5]{HornJohnson2013}.
Applying Cauchy--Schwarz inequality to \eqref{eq:traceless-pairing}, and then using \eqref{eq:kyfan-duality} and \eqref{eq:traceless-partial-traces}, we can prove \eqref{eq:d4}. This completes the proof.
\end{proof}

Activation shows that one-copy undistillability may be lost after an auxiliary entangled state is supplied. More precisely, every bipartite NPT state $\rho$ is either one-copy distillable or admits a PPT bound entangled state $\sigma$ such that $\rho\otimes\sigma$ is one-copy distillable \cite{klc02}. In particular, via a reduction to Werner states, such activator states can be chosen to be universal and arbitrarily close to the separable set for each fixed local dimension \cite{VollbrechtWolf2002Activation}. Consequently, if NPT bound entangled states exist, this phenomenon would imply that distillable entanglement is superadditive and nonconvex \cite{Shor2000NonadditivityOB}. Our next corollary provides a contrasting result, showing that two independent NPT Werner states cannot activate each other's one-copy distillability.

\begin{corollary}\label{thm:two-parameter}
Let $C\in\mathcal L(\bbC^m\otimes\bbC^n)$ satisfy $\rank C\le2$. For $\a,\b\in[-1/2,0]$, set
\begin{eqnarray}\label{eq:two-parameter}
 q_{\a,\b}^{(2)}(C)
 :=\norm C_F^2+\a\norm{\Tr_{\mathcal A}C}_F^2
 +\b\norm{\Tr_{\mathcal B}C}_F^2+\a\b|\Tr C|^2.
\end{eqnarray}
Then $q_{\a,\b}^{(2)}(C)\ge0$. Consequently, for $d_1,d_2\ge3$, $-1/2\le\a<-1/d_1$, and $-1/2\le\b<-1/d_2$, the NPT state $\rho_{\a,d_1}\otimes\rho_{\b,d_2}$ is one-copy undistillable.
\end{corollary}
\begin{proof}
Let $R=\ket{0}\bra{1}\in\mathcal L(\bbC^2)$. Applying Theorem~\ref{thm:main} to $C\otimes R$, regarded as an operator on $\mathcal A\otimes(\mathcal B\otimes\bbC^2)$, gives
$\norm{\Tr_{\mathcal A}C}_F^2\le2\norm C_F^2$. Similarly, it holds that $\norm{\Tr_{\mathcal B}C}_F^2\le2\norm C_F^2$. 
Hence both $q_{-1/2,0}^{(2)}(C)$ and $q_{0,-1/2}^{(2)}(C)$ are nonnegative, while Theorem~\ref{thm:main} gives $q_{-1/2,-1/2}^{(2)}(C)\ge0$. A direct expansion yields
\begin{eqnarray}\label{eq:two-parameter-interpolation}
 q_{\a,\b}^{(2)}(C)
 &=&(1+2\a)(1+2\b)q_{0,0}^{(2)}(C)
 -2\a(1+2\b)q_{-1/2,0}^{(2)}(C)\nonumber\\
 &-&2\b(1+2\a)q_{0,-1/2}^{(2)}(C)
 +4\a\b q_{-1/2,-1/2}^{(2)}(C)\ge 0.
\end{eqnarray}

For the two Werner states acting on
$\mathcal A_1\otimes\mathcal B_1$ and
$\mathcal A_2\otimes\mathcal B_2$ respectively, we consider the
bipartition $(\mathcal A_1\mathcal A_2):(\mathcal B_1\mathcal B_2)$. Let $\ket c$ be any vector of Schmidt
rank two across this bipartition. Under the vectorization in \eqref{eq:vectorization}, $\ket c$
corresponds to a matrix
$C\in \mathcal L(\mathcal A\otimes\mathcal B)$ with
$\dim\mathcal A=d_1$ and $\dim\mathcal B=d_2$.
A direct computation gives
\begin{eqnarray}
 \bra{c}{(\rho_{\a,d_1}\otimes\rho_{\b,d_2})^\G}\ket{c}
 &=&
 \frac{
  q_{\a,\b}^{(2)}(C)
 }
 {(d_1^2+\a d_1)(d_2^2+\b d_2)}\ge 0.
 \end{eqnarray}
This proves that $\rho_{\a,d_1}\otimes\rho_{\b,d_2}$ is one-copy undistillable. 
\end{proof}

Under the Choi--Jamio{\polishl}kowski isomorphism‌, Corollary~\ref{thm:two-parameter} equivalently states that the tensor product of the reduction-type maps $\Phi_{\a,d_1}(X):=\Tr(X)I_{d_1}+\a X$ and $\Phi_{\b,d_2}(X):=\Tr(X)I_{d_2}+\b X$ are $2$-positive. This is nontrivial because tensor products of $2$-positive maps need not remain $2$-positive in general, and such questions are closely related to NPT bound entanglement \cite{PhysRevA.61.062312,Clarisse2005PositiveMaps,PRXQuantum.3.010101,Stormer2010TensorPowers}.

\section{Outlook: the general $n$-copy problem}
Theorem~\ref{thm:costa-rico} converts the higher-copy distillability problem for Werner states into a family of rank-constrained partial-trace inequalities. As recalled in the Introduction, the unresolved one-copy-undistillable regime reduces to the endpoint $\alpha=-1/2$. Hence, for every $n\ge3$, the $n$-copy undistillability of the endpoint Werner state follows if
\begin{eqnarray}\label{eq:higher-copy-target}
 q^{(n)}_{-\frac12}(C)
 =\sum_{J\subseteq\{1,\ldots,n\}}
 \left(-\frac12\right)^{|J|}
 \norm{\Tr_J C}_F^2\ge0
\end{eqnarray}
holds for every $C\in\mathcal L(\mathcal H_1\otimes\cdots\otimes\mathcal H_n)$ with $\rank C\le2$. We first record that the rank-one case of \eqref{eq:higher-copy-target} is valid for every $n$.
\begin{proposition}\label{lem:rank-one-n-copy}
Let $\mathcal H=\mathcal H_1\otimes\cdots\otimes\mathcal H_n$ be a finite-dimensional Hilbert space. If $C\in\mathcal L(\mathcal H)$ has rank one, then
\begin{eqnarray}\label{eq:rank-one-n-copy}
 q^{(n)}_{-\frac12}(C)\ge 2^{-n}\norm C_F^2>0.
\end{eqnarray}
\end{proposition}
\begin{proof}
Let $C=\ket{x}\bra{y}$. Let $F_k$ exchange the two copies of the subsystem $\mathcal H_k$, and, for $J\subseteq\{1,\ldots,n\}$, put $F_J:=\prod_{k\in J}F_k$, with $F_\emptyset=I$. The swap trick gives
\begin{eqnarray}\label{eq:rank-one-swap-trick}
 \norm{\Tr_J C}_F^2
 =\ip{x\otimes y}{F_J(x\otimes y)}.
\end{eqnarray}
Consequently,
\begin{eqnarray}\label{eq:rank-one-product-form}
 q^{(n)}_{-\frac12}(C)
 =\ip{x\otimes y}{\displaystyle\prod_{k=1}^n\left(I-\frac12F_k\right)(x\otimes y)}.
\end{eqnarray}
The local swaps commute and $I-F_k/2\ge I/2$. It follows that
\begin{eqnarray}
 q^{(n)}_{-\frac12}(C)
 \ge 2^{-n}\norm{x}^2\norm{y}^2
 =2^{-n}\norm C_F^2.
\end{eqnarray}
This completes the proof.
\end{proof}

It remains to treat rank-two matrices. The two-copy proof does not extend directly because the simultaneous parity decomposition becomes more complicated. For $n=3$, the even-parity sectors are
\begin{eqnarray}
\label{eq:n3-even-parity}
 (+++),\qquad (+--),\qquad (-+-),\qquad (--+).
\end{eqnarray}
For general $n$, there are $2^{n-1}$ even-parity sectors, producing several coupled obstruction terms rather than the single term controlled by Lemma~\ref{lem:projection}. A proof for arbitrary $n$ should therefore control these sectors collectively in a spectator-stable way.

\vspace{0.8em}
\noindent\textit{Note added.---}After completing this work, we became aware of the very recent independent work of Fu, Gao, and Park~\cite{FuGaoPark2026}, which solves this two-copy Werner-state distillability problem through a variational Schur-complement argument. We believe that the two approaches provide complementary perspectives on this problem.

\section*{DECLARATION ON THE USE OF GENERATIVE AI}
The initial formulation and proof of Proposition~\ref{prop:three-square} were generated with GPT-5.5 Pro. The authors subsequently checked the argument in full, verified the results, and take full responsibility for the correctness of the manuscript.

\section*{ACKNOWLEDGMENTS}
The authors were supported by the NNSF of China (Grant No. 12471427).

\appendix
\section{Proof of Proposition~\ref{prop:three-square}}
\label{g}

For $X,E\in \mathcal M_{m,n}$, with corresponding vectors $\ket{x},\ket{e}\in\mathcal H$, we have
\begin{eqnarray}\label{eq:app-partial-traces}
 \Tr_{\mathcal B}(\ket{x}\bra{e})=XE^*,\qquad
 \Tr_{\mathcal A}(\ket{x}\bra{e})=(E^*X)^T.
\end{eqnarray}
Combining with \eqref{eq:qalpha} and \eqref{eq:two-column}, we obtain that
\begin{eqnarray}\label{eq:app-q-direct}
 q^{(2)}_{-1/2}(C)
 &=&\norm{X_1}_F^2+\norm{X_2}_F^2
 -\frac12\norm{E_1^*X_1+E_2^*X_2}_F^2\nonumber\\
 &-&\frac12\norm{X_1E_1^*+X_2E_2^*}_F^2
 +\frac14|\ip{e_1}{x_1}+\ip{e_2}{x_2}|^2.
\end{eqnarray}
Let
$
 G=[G_{ij}]:=(\ip{e_i}{x_j})_{i,j=1}^2.
$
We have
\begin{eqnarray}\label{eq:app-three-antisymmetric}
\lefteqn{2\norm{P_{--}(\ket{e_1}\otimes\ket{x_1})}^2
+\norm{P_{--}(\ket{e_1}\otimes\ket{x_2}+\ket{e_2}\otimes\ket{x_1})}^2
+2\norm{P_{--}(\ket{e_2}\otimes\ket{x_2})}^2}
\nonumber\\
&=&\frac14(3(\norm{X_1}_F^2+\norm{X_2}_F^2)
 -\sum_{i,j=1}^2\norm{E_i^*X_j}_F^2
 -\sum_{i,j=1}^2\norm{X_jE_i^*}_F^2\nonumber\\
&&+\norm G_F^2
 -\norm{E_1^*X_1+E_2^*X_2}_F^2
 -\norm{X_1E_1^*+X_2E_2^*}_F^2
 +|\Tr G|^2).
\end{eqnarray}
To see this, note that $\norm{P_{--}\ket{v}}^2=\ip{v}{P_{--}v}$, expand the formula for $P_{--}$ in \eqref{eq:projections}. It then suffices to apply Eqs.~\eqref{eq:swap-A}--\eqref{eq:swap-full} to the resulting inner products.

Consequently, a direct computation gives
\begin{eqnarray}\label{eq:app-PS-vector}
 P_{\Sym^2\mathcal S}\ket{w}
 =G_{12}\ket{e_1}\otimes\ket{e_1}
 +\frac{G_{22}-G_{11}}{2}
 (\ket{e_1}\otimes\ket{e_2}+\ket{e_2}\otimes\ket{e_1})
 -G_{21}\ket{e_2}\otimes\ket{e_2}.
\end{eqnarray}
Hence
\begin{eqnarray}\label{eq:app-PS-norm}
 \norm{P_{\Sym^2\mathcal S}\ket{w}}^2
 =\norm G_F^2-\frac12|\Tr G|^2.
\end{eqnarray}
Similarly, Eqs.~\eqref{eq:swap-A}--\eqref{eq:swap-full} and the formula for $P_{++}$ in \eqref{eq:projections} imply
\begin{eqnarray}\label{eq:app-Ppp}
 \norm{P_{++}\ket{w}}^2
 &=&\frac14(\norm{X_1}_F^2+\norm{X_2}_F^2
 +\sum_{i,j=1}^2\norm{E_i^*X_j}_F^2
 +\sum_{i,j=1}^2\norm{X_jE_i^*}_F^2\nonumber\\
 &+&\norm G_F^2-|\Tr G|^2
 -\norm{E_1^*X_1+E_2^*X_2}_F^2
 -\norm{X_1E_1^*+X_2E_2^*}_F^2).
\end{eqnarray}
Adding \eqref{eq:app-three-antisymmetric} and \eqref{eq:app-Ppp}, and then subtracting one half of \eqref{eq:app-PS-norm}, shows that the right-hand side of \eqref{eq:three-square} equals \eqref{eq:app-q-direct}. This proves \eqref{eq:three-square}.
$\hfill\square$

\bibliographystyle{unsrt}
\bibliography{references}

\end{document}